\documentclass[twoside,leqno,twocolumn]{article}
\usepackage[letterpaper]{geometry}

\usepackage{ltexpprt}

\usepackage[utf8]{inputenc}
\usepackage{amssymb}
\usepackage{amsmath}
\usepackage{mathtools}
\usepackage{color}
\usepackage{hyperref}
\usepackage{rotating}
\usepackage{pdflscape}
\usepackage{afterpage}
\usepackage{cite}

\usepackage{enumitem}
\setlist[enumerate]{noitemsep}
\setlist[itemize]{noitemsep}

\usepackage{libertine}

\title{Modular Subset Sum, \hspace{1in} \  \\
Dynamic Strings, \\ 
\hspace{1in} \  and Zero-Sum Sets}
\author{
Jean Cardinal\thanks{Université libre de Bruxelles. Supported by the Fonds de la Recherche Scientifique-FNRS under CDR Grant J.0146.18.}
\and 
John Iacono\thanks{Université libre de Bruxelles and New York University. Supported by NSF grants CCF-1533564 and by the Fonds de la Recherche Scientifique-FNRS under Grant n° MISU F 6001 1.}
}%
\date{}

\newcommand{\Mod}[1]{\ (\mathrm{mod}\ #1)}

\newtheorem{corr}{Corollary}[section]

\allowdisplaybreaks

\begin{document}
\setcounter{page}{1}

\fancyfoot[c]{\thepage}

\maketitle

\begin{abstract}
The modular subset sum problem consists of deciding, given a modulus $m$, a multiset $S$ of $n$ integers in $0..m-1$, and a target integer $t$, whether there exists a subset of $S$ with elements summing to $t \Mod{m}$, and to report such a set if it exists. 
We give a simple randomized with expected $O(m \log m)$ running time algorithm for the modular subset sum problem. 
This builds on and improves on a previous $O(m \log^7 m)$ w.h.p.~algorithm from Axiotis, Backurs, Jin, Tzamos, and Wu (SODA~19). 
Our method utilizes the ADT of the dynamic strings structure of Gawrychowski et~al.~(SODA~18).
However, as this structure is rather complicated  we present a much simpler alternative which we call the Data Dependent Tree.
As an application, we consider the computational version of a fundamental theorem in zero-sum Ramsey theory. 
The Erdős-Ginzburg-Ziv Theorem states that a multiset  of $2n - 1$ integers always contains a subset of cardinality exactly $n$ whose values sum to a multiple of $n$. 
We give an algorithm for finding such a subset in time $O(n \log n)$ w.h.p.~which improves on an $O(n^2)$ algorithm due to Del Lungo, Marini, and Mori (Disc.~Math.~09).

\end{abstract}

\section{Introduction}

In SODA 2019 \cite{DBLP:conf/soda/AxiotisBJTW19}, Axiotis, Backurs, Jin, Tzamos, and Wu gave an algorithm for modular subset sum with runtime $O(m \log^7 m)$ and that returns the correct answer with high probability.
This improved upon an earlier $\Tilde{O}(m^{5/4})$ algorithm of Koiliaris and Xu \cite{DBLP:journals/talg/KoiliarisX19} which first appeared in SODA 2017.

In~\S\ref{s:mss} we improve upon this with a very simple algorithm running in $O(m \log m)$ time in expectation. Our method is a straightforward implementation of the naïve dynamic programming approach, sped up using a recent data structure of Gawrychowski, Karczmarz, Kociumaka, Lacki, and Sankowski~\cite{DBLP:conf/soda/GawrychowskiKKL18}
that supports a number of operations on a persistent collection strings, notably, split and concatenate, as well as finding the longest common prefix (LCP) of two strings, all in at most logarithmic time with updates being with high probability. We also inform the reader of the independent work of Axiotis et al., \cite{DBLP:journals/corr/abs-2008-10577}, contains much of the same results and ideas for modular subset sum as here and appeared along with an earlier version of this paper in the proceedings of the 2021 SIAM Symposium on Simplicity in Algorithms (SOSA21). 

As the dynamic string data structure of~\cite{DBLP:conf/soda/GawrychowskiKKL18} is quite complex, we provide in~\S\ref{s:ddt} a new and far simpler alternative, which we call the \emph{Data Dependent Tree (DDT)} structure. Our general approach is to create a tree with the string stored in the leaves, and where the shape of the tree is a function of the data in the leaves and a random seed; in common with other simple string algorithms we use a hash function to compute fingerprints of strings, a method pioneered in by Karp and Rabin \cite{DBLP:journals/ibmrd/KarpR87}. The result is a structure with almost identical runtimes as \cite{DBLP:conf/soda/GawrychowskiKKL18} but which is only slightly more complex than a skip list and is easy to visualize (see Figure~\ref{fig}) and reason about. We say \emph{almost} identical as~\cite{DBLP:conf/soda/GawrychowskiKKL18} supports LCP queries in constant time whereas we do so in logarithmic time; this makes no overall difference in applications such as ours where the number of LCP queries do not asymptotically dominate the number of update operations, which take logarithmic time in both structures. Additionally, our runtimes hold in expectation in contrast to the stronger results of~\cite{DBLP:conf/soda/GawrychowskiKKL18} which hold with high probability\footnote{Earlier versions of this paper claim that our runtimes for the Data Dependent Tree hold with high probability. We retract such claims, as Lemma~\ref{l:fringe} only holds in expectation. We thank Pawel Gawrychowski and Tomasz Kociumaka for bringing this issue to our attention.}. 
None of the structures that for this problem that predate \cite{DBLP:conf/soda/GawrychowskiKKL18} 
(Sundar and Tarjan
\cite{DBLP:journals/siamcomp/SundarT94},
Melhorn, Sundar and Uhrig
\cite{DBLP:journals/algorithmica/MehlhornSU97}
and Alstrup, Brodal and Rauhe
\cite{DBLP:conf/soda/AlstrupBR00})
match the DDT's logarithmic time for all operations.

As an application, we consider the computational version of a fundamental theorem in zero-sum Ramsey theory   (see~\cite{CARO199693,Bialostocki1993,GAO2006337} for surveys). 
The Erdős-Ginzburg-Ziv Theorem \cite{GEZ} states that a multiset  of $2n - 1$ integers always contains a subset of cardinality exactly $n$ whose values sum to a multiple of $n$. 
We give an algorithm in \S\ref{s:zero} for finding such a subset in time $O(n \log n)$ w.h.p.~which improves on an $O(n^2)$ algorithm due to Del Lungo, Marini, and Mori \cite{DBLP:journals/dm/LungoMM09}.

% Our algorithm is thus also a $O(t \log t)$ algorithm for unrestricted (non-modular) subset-sum when the input is a multiset of $n$ positive integers at most $t$ in compact form and the target is $t$. Other leading solutions for subset-sun take time 
% $\Tilde{O}({n+t})$ \cite{DBLP:conf/soda/Bringmann17} with high probability.

\section{Modular subset sum via dynamic strings}
\label{s:mss}

The input to the modular subset sum problem is a positive integer modulus $m$, a multiset $S=s_1, s_2,\ldots , s_n$ of $n$ elements of $\mathbb{Z}_m$, and a target value $t$ in $\mathbb{Z}_m$. The multiset thus has at most $m$ distinct items and multiplicities are represented in compact form so that $S$ takes space $O(m)$.
The problem is to decide whether there exists a subset of $S$ whose sum of elements is congruent to $t \Mod{m}$. Our solution, in common with~\cite{DBLP:conf/soda/AxiotisBJTW19}, solves the problem for all values of $t$ in $\mathbb{Z}_m$ simultaneously.

\paragraph{Solution overview.}
Our solution is based on the classic dynamic programming approach, which we now describe.
We use the notation $S+x \coloneqq \{y + x \Mod{m} | y \in S\}$
and $[n]= \{1,2,\ldots,n\}$.
Let $S_i \subseteq \mathbb{Z}_m$ be the set of residues of all the sums of subsets of the first $i$ numbers of $S$: 
$$ S_i \coloneqq \left\{  \sum_{j \in X} s_j \Mod{m} \Bigg| X \subseteq [i] \right\}  $$
Given this definition, the problem is simply to determine if $t \in S_n$
We can construct $S_i$ recursively as follows:
$$ S_i =
\begin{cases}
\{0\} & \text{if $i=0$}\\
S_{i-1} \cup (S_{i-1} + s_i) & \text{otherwise.}
\end{cases} $$
If we wish to obtain the actual subset that adds to a given target $j$, call it $T_j$, as is typical for dynamic programming, another table is needed to record the choices made. Here it is sufficient to record for each target $j$ the index of the $S_i$ where it first was realized, which we call $a_j$:
$$a_j \coloneqq 
\begin{cases}
\displaystyle\min_{j\in S_i} i & \text{if }j \in S_n \\
\infty & \text{otherwise.}
\end{cases}$$ 
Given the $a_j$s, $T_j$ may be easily computed with at most $m-1$ recursions:
\begin{equation} \label{e:sets}
T_j =
\begin{cases}
\emptyset & \text{if $j=0$}
\\
\{s_{a_j}\} \cup T_{j-s_{a_j} \Mod{m}}  & 
\text{\parbox{0.8in}{if $j>0$ \\ \hspace*{6pt} and $a_j \not = \infty$}} 
\\
\textsc{No Subset} &
\text{otherwise.}
\end{cases}
\end{equation}

\paragraph{A string encoding.}
We encode the set $S_i$ in a binary string $\sigma_i \in \{ 0,1\}^m$, such that $j \in S_i$ if and only if the $j$th bit of $\sigma_i$, $\sigma_{i}[j]$, is $1$. The recurrence restated for strings is:
$$ \sigma_i =
\begin{cases}
1\overbrace{00\cdots0}^{m-1\text{ times}}  & \text{if $i=0$}\\
\sigma_{i-1} \vee (\sigma_{i-1} \gg s_i) & \text{otherwise.}
\end{cases} $$
where $\vee$ is the bitwise \textsc{or} operator and $\gg$ is a circular right shift.
In order to perform these two operations, we observe that they can be implemented directly using a recent data structure for dynamic strings:

\begin{theorem} \label{t:dynstrings}
\cite{DBLP:conf/soda/GawrychowskiKKL18}
There exists a data structure for maintaining a collection of strings and supporting the following operations, which we call the dynamic strings ADT:
\begin{itemize}
    \item $D \gets $ \textsc{New}$(x)$: Creates a new string containing a single character, $x$.
    \item $b\gets D$.\textsc{Equal}$(D')$: Returns whether $D$ and $D'$ are equal.    
    \item $\ell \gets D$.\textsc{Lcp}$(D')$: Returns the length of the longest common prefix of $D$ and $D'$.
    \item $D$.\textsc{Set}$(i,x)$: Sets the $i$th character of the string to $x$.
    \item $x \gets D$.\textsc{Get}$(i)$: Returns the $i$th character of the string.
    \item $D' \gets $D.\textsc{Split}$(i)$: The string beyond the $i$th character is removed from this structure and is placed in a new one, which is returned.
    \item $D$.\textsc{Concatenate}$(D')$: Concatenates the string represented by $D'$ to the end of $D'$, $D'$ becomes invalid.
\end{itemize}
The first three operations are constant time, the rest take time $O(\log n)$ where $n$ is the total size of strings stored, and the runtimes of update operations are in expectation.
\end{theorem}

Now observe that we can implement the circular shift operation using one
split and one concatenate operation, in time $O(\log m)$. 
Let $\sigma'_{i-1} \coloneqq (\sigma_{i-1} \gg s_i)$.
Finally, the bitwise disjunction $\sigma_{i-1} \vee \sigma'_{i-1}$ is implemented as follows: 
First find the index of the first bit that differs, which is the length of the LCP incremented by one. 
Then change this bit in one of the two strings so that they both have a 1 at this position. 
Then iterate until both strings are identical. 

There is a twist, however, in that after performing the splits and concatenates to implement the shifting of
$\sigma_{i-1}$ to obtain $\sigma_{i-1} \gg s_i$ we still need the data structure for $\sigma_{i-1}$ to compute $\sigma_{i-1} \vee (\sigma_{i-1} \gg s_i)$, and we cannot afford to make a copy. 
Fortunately, the technique of persistence exists exactly for this purpose. In its simplest form, known as partial persistence, read-only access to previous versions of a data structure are supported. The dynamic string structure of \cite{DBLP:conf/soda/GawrychowskiKKL18} supports persistent access to the strings, and our alternative DDT structure can be made to be partially persistent via the general transformation of \cite{DBLP:journals/jcss/DriscollSST89}. 

To summarize:

\begin{theorem}
Given an integer modulus $m>0$, a target $t$, and a compact representation of a multiset $S$ of integers in $\mathbb{Z}_m$, the following algorithm determines if there is a subset of $S$ that sums to $t$ in  expected $O(m \log m)$ time, and reports it: 
\begin{enumerate}
    \item Initialize a string data structure $D$ with the string $1\overbrace{00\cdots0}^{m-1\text{ times}}$
    \item \textbf{For each} value $s_i$ with multiplicity $\mu_i$ in $S$: \label{outer}
    \item Initialize $A \gets [\emptyset,\overbrace{--\cdots-}^{m-1\text{ times}}]$
    \begin{enumerate}
        \item Use persistence to save $D$ as read-only structure $D_{\text{1}}$.
        \item Set $count \gets 0$
        \item \textbf{Do:} \label{mid}
        \begin{enumerate}
            \item Use persistence to save $D$ as a read-only structure $D_{\text{2}}$.
            \item Do a circular rotation of $D$ by $s_i$ using one split and one concatenate.
            \item \textbf{While} $\textbf{not }D.\textsc{Equal}(D_{\text{2}})$ \label{inner}
            \begin{enumerate}
                \item Let $k \gets D.LCP(D_{\text{2}})+1$
                \item If $D.\textsc{Get}[k]=0$:  \\ \hspace*{8pt} $D.\textsc{Set}(k,1)$ 
                 \\ \hspace*{8pt} 
                $A[k]\gets s_i$
                \item If $D_{2}.\textsc{Get}[k]=0$:  \\ \hspace*{8pt} 
                $D_{2}.\textsc{Set}(k,1)$
            \end{enumerate}
            \item Increment $count$
        \end{enumerate}
         \textbf{until} $count = \mu_i$ or $D.\textsc{Equal}(D_{\text{1}})$.
    \end{enumerate}
\item If the $t$th character of the string stored in $D$ is $1$ there is a subset. Use Equation~\ref{e:sets} to report it using $A$
\end{enumerate}
\end{theorem}

\begin{proof}
The correctness of the algorithm follows from the earlier discussion, we focus on the runtime.
For the runtime of the outermost loop, recall that $n$ is the total size of the multiset $S$, but that the total number of distinct items in $S$ is at most $m$, and we have placed no restriction on the relationship of $m$ and $n$. 
As such, the outer loop, step~\ref{outer}, runs at most $m$ times. 

For the mid-level loop, step~\ref{mid}, a simple analysis says that this is run at most once for each element in $S$ and thus is run at most $n$ times total. However, as observed by \cite{DBLP:conf/soda/GawrychowskiKKL18}, if this loop runs without changing $D$, then subsequent iterations of the loop with the same value of $s_i$ will also not change $D$ and can be skipped. Thus the loop is only run once for each distinct value in $S$, of which there are $m$, plus once for every time $D$ changes. This can happen at most $m-1$ times as $D$ begins with $m-1$ zero-bits and the only change is to make a bit one.

For the runtime of the innermost loop, step~\ref{inner} we need a simple observation, where $\ominus$ is the symmetric difference of the sets (\textsc{Xor}).
\begin{equation} \label{triangle}
|S_{i-1} \ominus (S_{i-1}+s_i)| = 2 \cdot |S_i \setminus S_{i-1}|.
\end{equation}
Summing \eqref{triangle} yields:
$$
\sum_{i=2}^{n}|S_i \ominus (S_{i-1}+s_i)|
= \sum_ {i=2}^{n} 2|S_i \setminus S_{i-1}| \leq 2|S_n| \leq 2m.
$$
The sum, $\sum_{i=2}^{n}|S_{i-1} \ominus (S_{i-1}+s_i)|$, is exactly the number of times a differing bit will be detected and corrected in step~\ref{inner}, and thus this loop will run in at most $2m$ times total. 
For the computing of the contents of the set, the recursion from~\eqref{e:sets} runs in time $O(m)$ by construction.

In summary, each line of the psudocode will run at most $O(m)$ times, and each line possibly calls the string structure, at a cost of $O(\log m)$ w.h.p.~per call, for a total runtime of $O(m \log m)$ w.h.p. 
\end{proof}

\section{Zero-sum sets of prescribed size}
\label{s:zero}

Every sequence of $n$ elements of $\mathbb{Z}_n$ has a zero-sum contiguous subsequence. This is a standard application of the pigeonhole principle: define the partial sums $s_i$ of the first $i$ elements, and observe that if none of them is $0$, then two of them must be equal. What is the computational complexity of finding such a subsequence? In this case, it is simply a matter of solving the element distinctness problem, which can be done in time $O(n \log n)$.
A fundamental result of zero-sum Ramsey theory is the existence of zero-sum subsets in of \emph{prescribed cardinalities} in any multiset of $2n - 1$ items in $\mathbb{Z}_n$:

\begin{theorem}
(Erdős-Ginzburg-Ziv \cite{GEZ})
Every multiset of $2n - 1$ elements of $\mathbb{Z}_n$ contains a subset of exactly $n$ elements the sum of which is $0 \pmod{n}$.
\end{theorem}
We consider the computational EGZ problem: Given a sequence of $2n - 1$ elements of $\mathbb{Z}_n$, identify $n$ of them that sum to $0$. This is a search problem, since from the EGZ Theorem the decision problem is trivial.
We present an algorithm that runs in time $O(n \log n)$ w.h.p. This improves on the $O(n^2)$ algorithm from Del Lungo, Marini, and Mori \cite{DBLP:journals/dm/LungoMM09}. The algorithm is based on the original proof of the Theorem, itself based on a version of the Cauchy-Davenport Theorem. See Alon and Dubiner for a discussion of various alternative proofs \cite{Alon93zero-sumsets}.

\subsection{Original proof of the EGZ Theorem}
First we state the Cauchy---Davenport theorem:

\begin{theorem}
(Cauchy~\cite{c}---Davenport~\cite{10.1112/jlms/s1-10.37.30}).
Let $p > 2$ be a prime and $B = b_1, b_2, \ldots , b_{p-1}$ a multiset of $p-1$ nonzero elements of $\mathbb{Z}_p$. Then, every element of $\mathbb{Z}_p$ %\john{was $\mathbb{Z}_n$} 
can be written as the sum of elements of a subset of $B$.
\end{theorem}

%\begin{proof}
\paragraph{When $n$ is prime.} 
We first prove the EGZ theorem for $n=p$, with $p$ prime. 
We arrange the $2p-1$ residues $a_1 \leq a_2 \leq \cdots \leq a_{2p-1}$ in increasing order. 
First suppose that $a_i=a_{i+1}= \cdots = a_{i+p-1}$ for some $i$. Then $\sum_{j=i}^{i+p-1}a_j = pa_i \equiv 0 \Mod{p}$.
Hence we can assume that $a_i \not = a_{i+p-1}$ for all $i \leq p$.

Let $\sum_{i=1}^p a_i \equiv c \not \equiv 0 \Mod{p}$ and $B=\{ b_i \coloneqq a_{i+p} - a_{i+1}| 1 \leq i \leq p-1\}$. 
From our assumption, all the $b_i$ are nonzero, hence we can apply the Cauchy-Davenport Theorem to show that the equation $\sum_{b_i \in B'} b_i \equiv -c \Mod{p}$ for some $B' \subseteq B$ has a solution. For this solution, we obtain the expression
$$
\sum_{i=1}^p a_i + \sum_{b_i \in B'} b_i   \equiv \sum_{i=1}^p a_i + \sum_{b_i \in B'} ( a_{i+p} - a_{i+1}) \equiv 0 \pmod{p}
$$
which is a sum of $p$ elements.

\paragraph{In general.} In order to prove the statement for any $n$, we consider the case where $n = uv$ for two integers $u$, $v > 1$.
By induction, in any sequence of numbers $a_1, a_2, \ldots , a_{2uv-1}$, there are $u$ of them summing to a multiple of $u$. 
We can remove them and repeat $2v - 1$ times, to obtain $2v - 1$ disjoint subsequences of $u$ numbers, each summing to a multiple of $u$. 
Let $c_iu$ denote the sum of the $i$th subsequence, for $1 \leq i \leq 2v - 1$. 
Now again by induction, in the sequence of numbers $c_1, c_2, \ldots , c_{2v-1}$, there are $v$ of them summing to a multiple of $v$. Summing all the $a_i$ in each of these $v$ subsequences, we obtain $uv$ numbers summing to a multiple of $uv$, as wished.

\subsection{Algorithm}

We first consider the case where $n = p$ for some prime $p$. We can compute the value $c$ and the $b_i$ in linear time. We then solve the equation $\sum_{b_i \in B'} b_i \equiv -c \Mod{p}$, $B' \subseteq B$, for $B'$, which is an instance of the modular subset sum problem,
solvable in time $O(p \log p)$ w.h.p.

It remains to consider the case where $n$ is composite. 
Let $u$ be the largest prime factor of $n$, and $v = n/u$. 
The algorithm first runs $2v - 1$ times the algorithm for the prime case with an input of size $2u - 1$, then recurses once on an input of size $2v - 1$. 
The first step takes time $(2v - 1)O(u \log u) = O(n \log u)$ w.h.p. The overall running time is therefore $f (n) \leq O(n \log u) + f (n/u) = O(n \log n)$ w.h.p. Note that we need to find the prime factors of $n$, which can be obtained in time sublinear in $n$, as there are at most $\log n$ prime factors of size at most $\sqrt{n}$.

\section{Data Dependent Trees} \label{s:ddt}
\subsection{Motivation}

Our simple method for modular subset sum depends heavily on the string structure of~\cite{DBLP:conf/soda/GawrychowskiKKL18}.
A fair question is are we simply masking the complexity of our solution by using this complicated string structure as a black box? In \cite{DBLP:journals/corr/GawrychowskiKKL15}, before they delve into the details of their structure they note the following:

\begin{quote}
We note that it is very simple to achieve $O(\log n)$ update time for maintaining a non-persistent family of strings under concatenate and split operations, if we allow  the equality queries to give an incorrect result with polynomially small probability.
We represent every string by a logarithmic-height tree with characters of the string in the leaves and where every node stores a fingerprint of the substring of the sequance represented by its descendant leaves.
However, it is not clear how to make the answers always correct in this approach (even if the time bounds should only hold in expectation). Furthermore, it seems that both computing the longest common prefix of two strings of length $n$ and comparing them lexicographically requires $\Omega(\log^2 n)$ time in this approach. 
This is a serious drawback, as the lexicographic comparison is a crucial ingredient in our applications related to pattern matching.
\end{quote}

We use this simple fingerprinting approach to create a data structure for the dynamic string problem but without the deficiencies mentioned above. 
We can eliminate the chance of wrong equality queries by ensuring that nodes have the same fingerprint if and only if they induce  the same string in their leaves.
More challenging, however, it to reduce the runtime of LCP queries to $O(\log n)$.
Our key realization is this can be done easily if trees/subtrees with the same induced string in their leaves always have the same shape, we call a structure with this property \emph{data dependent}.
Typical balanced tree structures certainly are not data dependent.
We thus present a new tree structure, which is data dependent and which we call a \emph{data dependent tree (DDT)}.
Our structure draws inspiration from the literature on data structures with unique representations \cite{DBLP:journals/siamcomp/SundarT94,DBLP:journals/siamcomp/AnderssonO95,DBLP:conf/swat/BlellochGV08,DBLP:conf/icalp/Golovin09,DBLP:journals/corr/abs-1005-0662}, especially skip lists \cite{DBLP:journals/cacm/Pugh90} and treaps \cite{DBLP:journals/algorithmica/SeidelA96}, but where the source of randomness is the derived from the fingerprints of the data so that parts of the tree with the same data are constructed identically.
Our tree is not binary, as no tree can be binary, have the data dependent properly, and support concatenate in time $o(\sqrt{n})$ \cite{DBLP:conf/focs/Snyder77}.
Our structure, as in previous work, assumes an oblivious adversary.

Our structure can be substituted for that of \cite{DBLP:conf/soda/GawrychowskiKKL18}, and is thus of independent interest. We refer the reader to \cite{DBLP:conf/soda/GawrychowskiKKL18} for a full discussion of the applications of a dynamic string structure. As noted in the introduction, the only difference in runtime is that our LCP queries take time $O(\log n)$ rather than constant, which would only matter if these operations asymptotically dominate updates, and that our runtimes are expected rather than with high probability. In~\cite{DBLP:conf/soda/GawrychowskiKKL18}, the runtime for LCP queries is initially logarithmic and is then sped up to constant with an auxiliary structure which utilizes the constant-query-time dynamic least common ancestor (LCA) query structure of Cole and Hariharan \cite{DBLP:journals/siamcomp/ColeH05}; we avoid such complication here.

\subsection{Description}

The structure is a rooted tree.
We describe it bottom-up, level by level starting from the leaves. All nodes in the tree have a $c \log n$-bit fingerprint, which will be equal if their subtrees contain equal data and we will ensure are different if they are different. We assume a random hash function $h$ to compute the fingerprint.

See Figure~\ref{fig} for an example of the construction we now describe.
The leaf level of the tree contains the string, one character per leaf. 
The hash of a leaf containing character $x$ is $h(x)$. 
The levels of the tree then alternate between \emph{duplicate levels} and \emph{increasing levels}.
If level $\ell$ is a duplicate level, each node (\emph{a duplicate node}) has as children a maximal consecutive set of nodes from level $\ell-1$ which have the same hash.
If level $\ell$ is an increasing level, each node (\emph{a increasing node}) has as children a maximal consecutive set of nodes from level $\ell-1$ which have increasing hashes, left-to-right. 
(We note that a grouping similar to what we propose for increasing levels appeared in \cite{DBLP:journals/algorithmica/MehlhornSU97}, however their approach differs form here in that 
the source of randomness used for the grouping is independent of the data and so efficient LCP queries were not possible).
The hash of an increasing node node of level $\ell$ with children $c_1, c_2, \cdots c_k$ is a hash of the level number and the hashes of the children: $h(\langle \ell, h(c_1), h(c_2), \ldots h(c_k) \rangle)$.
The hash of a duplicate node of level $\ell$ with children $c_1, c_2, \cdots c_k$ is a hash of the level, the number of children, and the hash of one of the children (the hash of all children is identical), $h(\langle \ell, k, h(c_1) \rangle)$.

All nodes are also augmented with the number of leaves in their subtree and the nodes of the tree are level-linked.
We use \emph{neighbor} to refer to an adjacent node in a level and \emph{sibling} to refer to a node that shares the same parent. 
As the tree is not binary and there is no bound on the number if children of a node, we need to describe the data structure a node uses to store its children.
The children of increasing nodes are stored in a linked list. 
The children of duplicate nodes are stored in a search tree that supports the standard operations as well as split and concatenate in time $O(\log n)$, as well as access, insertion, and deletion of the minimal and maximal elements in time $O(1)$ amortized. A 2-3-4 tree with pointers to the first and last elements meets these requirements (see \cite{DBLP:books/el/Leeuwen90}).
Note that when we refer to the shape our structure, we are speaking of the global tree only as visualized in Figure~\ref{fig} and do not take into account the secondary structures to efficiently store the children, which are not shown in Figure~\ref{fig}.

\paragraph{Collisions and height: rebuilding}
So far as described, if two nodes have different hashes they represent different strings, but different strings could by chance hash to the same value. To eliminate this one-sided error, we explicitly check for it.
We maintain a hash table, the \emph{fingerprint table}, containing all hashes used, and the inputs used to calculate each, e.g., for an increasing node, the level number and the hashes of the children. 
When we create a new node, we compute the hash, and check to see if the hash has been used before. 
If it was, we verify that the parameters used to calculate it were the same.
If, however the new node shares the hash of an existing node, but had different inputs, then there is a collision, a new random hash function is chosen, and the structure is rebuilt from scratch. By using a hash function that generates a $c \log n$ bit integer collisions happen with probability at most $\frac{1}{n^c}$ per new hash, which is certainly at most $\frac{1}{n^{c-1}}$ per operation, and thus the cost to check and rebuild if needed is constant expected amortized when $c\geq 3$.

Additionally, we show in Lemma~\ref{l:height} that the height of the tree is $O(\log n)$ with high probability; we choose an $\alpha$ such that if a height of $\alpha \log n$ is surpassed the structure is rebuilt; with $\alpha
$ at least 5 this happens with probability at most $\frac{1}{n^2}$ and thus the amortized expected cost to rebuild is constant for each operation.

\subsection{Properties}

We call the left ID of a node to be the index of the leftmost leaf in its subtree, and define the right ID similarly.
If we view nodes as being drawn with $x$ coordinates based on the left ID, one discovers that our structure bears close resemblance to a skip list \cite{DBLP:journals/cacm/Pugh90}. Here we prove several properties of our structure that will be of use in the analysis of the operations. In skip lists a non-leftmost node is promoted to the next level with probability $1/2$, and we begin by stating the similar fact that in our structure a non-leftmost node as identified by its left ID survives two levels up with a probability at most $\frac{1}{2}$: 

\begin{fact}\label{f:gp}
Given a leaf $n$ with a given left ID, the only nodes that can share $n$'s left ID are those ancestors of $n$ that have $n$ as their leftmost child; these are a connected path starting from $n$.
Given an increasing node $n$ with parent $p$ and grandparent $g$, if $n$ has the same hash as its neighbor to the left, it will not be the left child of its parent, and thus it will not share a left ID with its grandparent. Otherwise, whether $n$'s parent's hash is larger or smaller than its left neighbor determines if $n$'s grandparent will have the same left ID as $n$.
Thus the chance that a non-leftmost $n$ and its grandparent share the same left ID is either zero or a probability $\frac{1}{2}$ event, depending on whether $n$ has the same hash as its left sibling.
\end{fact}

Given this fact we can now bound the height of our structure. The proof is basically the same as for skip lists.

\begin{lemma} \label{l:height}
The height of the DDT is at most $2(c+1) \log n$ with probability $1-\frac{1}{n^{c}}$
\end{lemma}

\begin{proof}
Let $x$ be a leaf. Let $l(i,\ell)$ be the event that there is a non-leftmost node at level $\ell$ with the same left ID as the $i$th leaf. Given $l(i,\ell)$ is true, for some duplicate level $\ell$, $l(i,\ell+2)$ is true with probability at most $\frac{1}{2}$ by Fact~\ref{f:gp}.
Thus $Pr[l(x,\ell)]$ is at most $\frac{1}{2^{\lfloor \ell/2 \rfloor}}$.

What is $Pr[\bigcup_{i=1}^n l(i,2c\log n)  ]$? From the union bound, we have $\Pr [\bigcup_{i=1}^n l(i,2c\log n)] \leq 1/n^{c-1}$. Note that the events $Pr[l(i,\ell)]$ are \emph{not} independent for each $i$, as they are for skip lists, but that is ok as the union bound does not require independence. 
\end{proof}

One difference with skip lists is that for $k$ nodes with differing hashes to have an increasing parent they must be in increasing order, which happens with probability $\frac{1}{k!}$, whereas in skip list this happens with probability $\frac{1}{2^k}$. This complicates the direct application of Chernoff bounds to compute the length of a search path, as is done in skip lists. We need the following lemma which shows how Chernoff bounds can be used for a limited type of dependent event:

\begin{lemma} \label{l:rps}
If one flips a fair coin until $\log n$ heads appear, this takes more than $c \log n$ flips with probability at most $n^{-c/16}$. 
Suppose instead of using a fair coin, before each flip an adversary can choose an arbitrary biased coin with probability of heads at least $\frac{1}{2}$, and this choice can depend on the results of coin flips so far. Then, the number of biased coin flips needed until $\log n$ heads appear is greater than $c \log n$ with probability at most $n^{-c/16}$.
\end{lemma}

\begin{proof}
The statement about fair coins is textbook application of Chernoff bounds.
For the statement about biased coins, when choosing a coin of probability $p \geq {1}{2}$, we instead imagine an event with three outcomes: rock with probability $\frac{1}{2}$, paper with probability  $p-\frac{1}{2}$ and scissors with probability $1-p$. 
Thus rock-or-paper corresponds to a biased head and scissors corresponds to a biased tails.
As rock occurs in each event independently with probability $\frac{1}{2}$, Chernoff applies and thus we know the probability of more than $c \log n$ trials to obtain $\log n$ rocks is less than $n^{-c/16}$. 
The probability of more than $c \log n$ trials to have $\log n$ rocks-or-paper events is thus less than the probability of more than than $c \log n$ events to have $\log n$ rock events and thus is less than $n^{-c/16}$.
\end{proof}

\begin{fact} \label{f:different}
Given two nodes in the same structure, or in different DDTs constructed with the same hash function, their induced strings are identical if and only if the hashes of the nodes are equal.
\end{fact}

\begin{lemma} \label{l:inc}
The number of increasing nodes on the path from the root to the $i$th leaf is $O(\log n)$ w.h.p.
\end{lemma}

\begin{proof}
Consider walking backwards from the $i$th leaf to the root. At each step we either go left or up. We know the height is bounded by $\alpha \log n$, which bounds the number of times we move up. We view the choice to go left or up in the following cases, where $x$ denotes the current node, and $I$ denotes $x$ and its right siblings, which form an increasing sequence. If we are at the leftmost node on a level, we always go up. If the node to the left of $x$ is has the the same hash as a node in $I$, then we go up. Otherwise, we compare the hash of the node to the left of $x$ to the hash of $x$, if it is smaller we go left and if it is larger we go up. The chance that it is smaller that all elements of $I$ is $\frac{1}{|I|+1}$, which is at most $\frac{1}{2}$. Applying Lemma~\ref{l:rps}, we obtain the result.
\end{proof}

\begin{lemma} \label{l:dup}
The sum of the logarithms of the degrees of the duplicate nodes on the path from the root to a leaf is $O(\log n)$.
\end{lemma}

\begin{proof}
Let $d_1, \ldots d_k$  be the duplicate nodes on a root-to-leaf path. Let $c(d_i)$ be the number of children of $d_i$ and let $s(d_i)$ be the number of leaves in the induced subtree of $d_i$. We know that for every child $m$ of $d_i$, $s(m)=\frac{s(d_i)}{c(d_i)}$, as in a duplicate node, all children have identical hashes and thus by the first point represent identical strings. Thus 
$s(d_{i+1})\leq \frac{s(d_i)}{c(d_i)}$, and
$$ s(d_k)\geq \prod_{i=1}^k c(d_i) = 2^{\log \prod_{i=1}^k c(d_i)}= 2^{\sum_{i=1}^k \log c(d_i)}.
$$
Since $n \geq s(d_{i+1})$, we have 
$n \geq 2^{\sum_{i=1}^k \log c(d_i)}$, or equivalently
$\log n \geq \sum_{i=1}^k \log c(d_i)$.
\end{proof}

From Lemmata~\ref{l:inc} and \ref{l:dup} we have bounded both types of levels and can combine:

\begin{corr}\label{c:traverse}
Traversing the DDT from the root to the $i$th leaf takes time $O(\log n)$ w.h.p.
\end{corr}

For split and merge, as in skip lists, we need the following lemma which will bound the cost of the partial rebuilding required:

\begin{lemma} \label{l:fringe}
Given any set of at most $c \log n$ nodes $S$ with distinct hash values, none of which are ancestors or descendants of each other,
the total number of nodes that are ancestors of elements of $S$ and share a left ID with an element of $S$ is $O(\log n)$ in expectation.
\end{lemma}

\begin{proof}
We proceed through each item in $S$, walking up in the structure while the left ID remains the same and moving on to the next item in $S$ if it does not. As Fact~\ref{f:gp} bounds the chance of going up as at most $\frac{1}{2}$ we can thus apply Lemma~\ref{l:rps} to complete the proof.
\end{proof}

\subsection{Operations}

We describe in detail the DDT's implementation and analysis of the dynamic strings ADT operations. Given a hash function, the shape of DDT is completely determined by the data, and thus the implementation of the modifying operations is a simple matter of removing parts that become invalid and rebuilding bottom-up.

As mentioned in the description, any collision of the hashes, or having the height exceed $\alpha \log n$ will trigger an immediate rebuilding of the structures of all strings in the collection with a new hash function. 

\begin{description}
\item[Get.] To obtain a character we navigate to its leaf guided by the augmented subtree sizes, this takes time $O(\log n)$ w.h.p.~by the Corollary \ref{c:traverse}.
\item[Equal.] Simply compare the hashes of the roots in constant time, this works by Fact~\ref{f:different}.

\item[Longest Common Prefix (LCP).]
See Figure~\ref{fig2}.
The method is simple, start two pointers, $p$ and $q$, at the root of the shallower of the two trees, and the leftmost node at the same level in the other tree, and move them down so that they always point to the first node on that level where the sequences of hashes at that level differ; when they reach the leaves they will be at the first differing character in the strings. 
At any step $p$ and $q$ are simply moved to the first different child. The only special case is if w.l.o.g.~$p$ has $k$ children, $q$ has $>k$ children, and the first $k$ children of both have the same sequence of hashes. Then $q$ is moved to its $k+1$th child and $p$ is moved to its right neighbor's leftmost child. 

For the runtime, only $O(\log n)$ nodes are traversed by the pointers so we only need to bound the cost of moving the pointers to their children. It is $O(\log n)$ by Lemma~\ref{l:inc} which bounds the time to move to children of increasing nodes, combined with the observation that all identical nodes where one moves the a pointer to a non-minimal/maximal element (for which we have direct pointers) are on a root-to-leaf path and thus  Lemma~\ref{l:dup} applies.

\item[Split.] See Figure~\ref{fig3}.
The first step is to split all nodes that are ancestors of the split point as follows: take all such nodes and remove them and replace them with up to two new nodes, one which contains all former children whose descendent are entirely to the left of the split point, and one symmetrically for the right. One or both of these sets of children may be empty in which case no node is created for that set.
The nodes that are created are parentless.

The left and right halves of the structure are now disconnected, and we now discuss how to rebuild the left half, the right half is symmetric. 
We proceed bottom up and rebuild in the obvious way. For example, to build a increasing level $\ell$, we process the nodes of level $\ell-1$ that have no parent by adding them to the rightmost node of level $\ell$ if their hash is larger than their left neighbor's hash, or adding a new node to level $\ell$ as a parent if their hash is smaller than their left neighbor's hash. If the former rightmost node of level $\ell$ has children added to it, its hash must be recomputed and all of its ancestors have their parent pointer detached as with the new hash it is now unknown if they are allowed to be attached to their parent. All nodes that have no children as a result of this are removed.

For the runtime, the splitting process splits $O(\log n)$ nodes and has the same runtime as traversal, $O(\log n)$ w.h.p. from Corollary~\ref{c:traverse}. For the rebuilding phase, the amortized time is linear in the number of nodes touched in the rebuilding process. (Recall that inserting a new minimum or maximum element into the BST of a duplicate node takes time $O(1)$ amortized). 
Call the fringe nodes the leftmost nodes on each level that are either new or have had a child added. All new and changed nodes are part of the fringe or are to the right. 
But, these newly added nodes share a left ID with one of the $O(\log n)$ nodes on the fringe, an thus by applying Lemma~\ref{l:fringe} to the top node with each left ID on the fringe, there are expected to be $O(\log n)$ such nodes. Building these nodes takes time linear in total number of number of children, which sums to $O(\log n)$ in expectation for newly added nodes with more than one child as all newly added nodes have a fringe node in every subtree.

The last aspect of the runtime to consider are the nodes that may have been removed as part of the rebuilding process as they only had one child that they were detached from. We simply use amortization to make such deletions free by paying for the deletion in advance when the node is created.

\item[Concatenate.] First we proceed from the bottom up to see if the leftmost node of the right part and the rightmost node of the left part can be merged into one node. If such a merge can be done, all of the ancestors of the nodes to be merged are detached from their parents and the merge is performed. Then from that level up, a rebuilding phase commences, as in split with one minor difference: in split for the left structure at each level we began with the leftmost existing node, added children if possible, and then added new nodes to the right. Here we do the same but at the end check to see if the rightmost node added (or the rightmost node of the former left if none were added) can be merged with the leftmost node of the former right, if so this is done and once again the all ancestors are detached from their parent.

The analysis is as for split, except we need to define a right and left fringe, and all added nodes have the same left ID as a node on the left fringe or the same right ID as a node on the right fringe.
    
\item[Set.] Changing a character in the string can implemented by two splits, a new, and a concatenate.

\end{description}

We summarize the runtimes for the operations above:

\begin{theorem}
Data dependent trees support the dynamic strings ADT as presented in Theorem~\ref{t:dynstrings}. Equality testing takes time $O(1)$ and the remaining operations take time $O(\log n)$, with the runtime of the update operations being expected and amortized, where $n$ is the total size of all strings stored.
\end{theorem}

In our modular subset sum algorithm, we needed the dynamic strings ADT to be partially persistent. 
We note that as our data structure, excluding the fingerprint collision table, is a pointer-based data structure of constant degree and thus partial persistence \cite{DBLP:journals/jcss/DriscollSST89} can be applied to be able to execute queries on old versions of the structure. 
The more complex confluent persistence \cite{DBLP:journals/jal/FiatK03,DBLP:conf/soda/ColletteIL12} is not needed, as concatenate operations are not performed across strings of different versions.
The fingerprint collision table is not needed to answer queries and thus persistence need not be applied to it. Care must be taken however as the LCP and equality queries only work when the different structures were constructed with the same hash value. Thus if a tree was rebuilt  with a new hash function because of a hash collision or excessive height, queries that involve versions before and after the rebuild must be forbidden. The easiest way to ensure this is if the DDT needs to be rebuilt, the entire algorithm that uses it restarts from scratch. So long as the algorithm that uses it has runtime bounded by a known polynomial in the maximum size of the DDT, as is the case with modular subset sum, the rebuild and restarting costs do not asymptotically affect the expected amortized runtime.

\begin{landscape}
\begin{figure}
    \centering
    \includegraphics[width=9in]{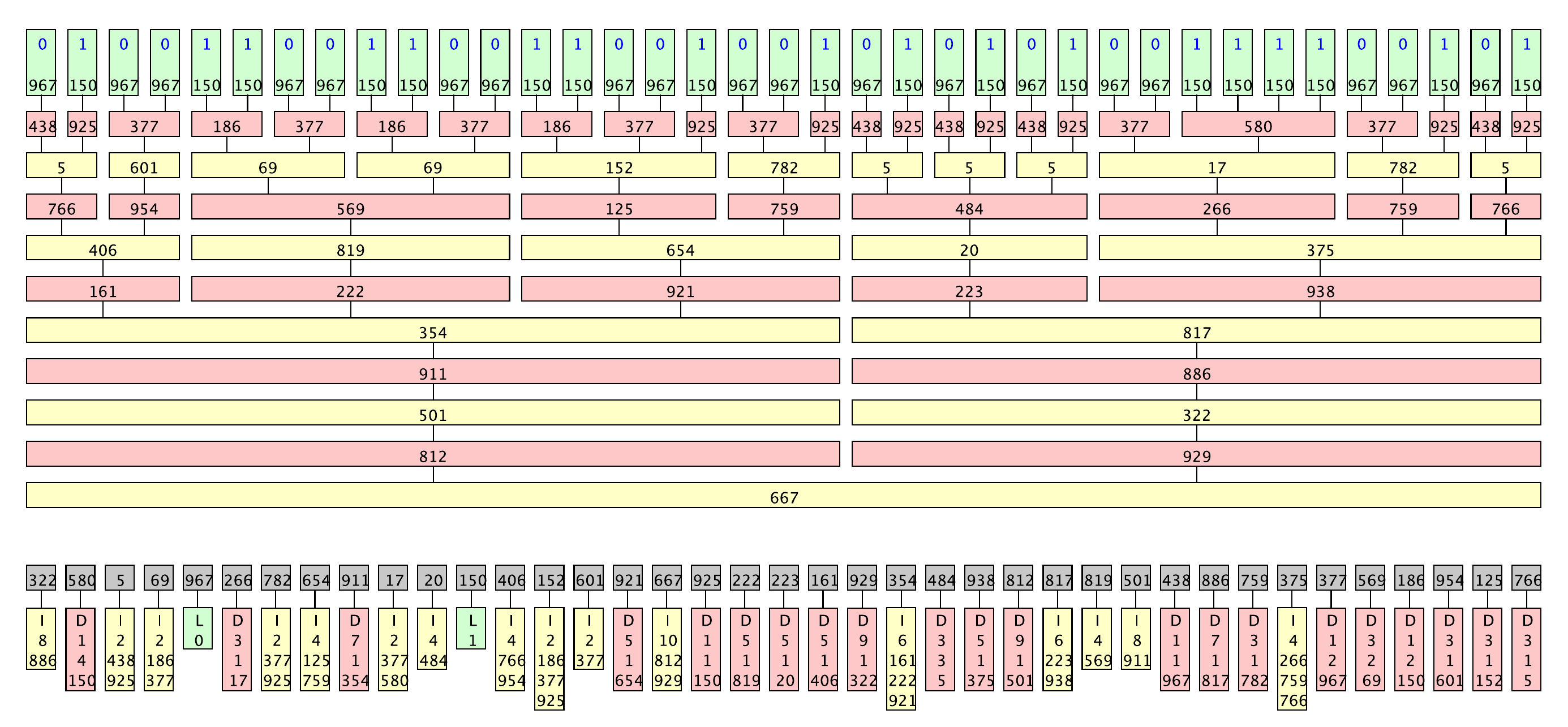}
    \caption{Illustration of the data dependent tree (DDT), storing the binary string at the top. A random hash function to the integers $[1000]$ was chosen for simplicity. Yellow nodes are increasing nodes, whose children have increasing hashes left-to-right, and red nodes are duplicate nodes, whose children have identical hashes. The extent of each node visualises the range from the left ID to the right ID of the node. At the bottom is the fingerprint hash table, where, for example, the second entry, 580, indicates that it is a duplicate node from the first level and has four children with a hash of 150 each. We note that this figure was made programmatically by implementing the construction of the DDT; this is just a screen of code in Processing's python mode.}
    \label{fig}
\end{figure}
\end{landscape}

\begin{landscape}
\begin{figure}
    \centering
    \includegraphics[width=8.5in]{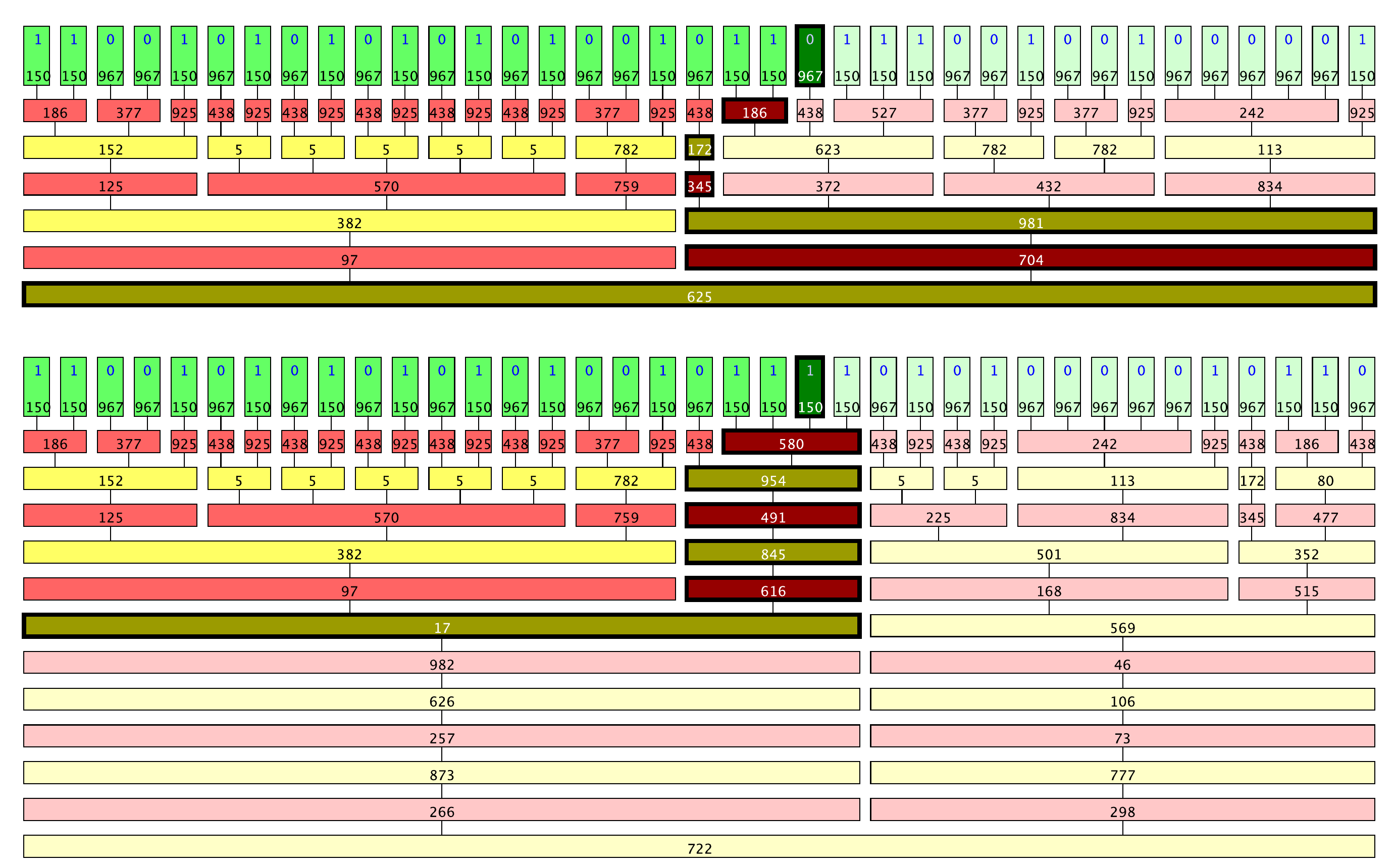}
    \caption{Longest common prefix. Dark nodes represent the search path, which starts from the leftmost node of the highest level the two structures have in common. In general the search moves to the leftmost differing child. The only special case is when the children of one are a prefix of the children of the other. This occurs in the figure where 186 (top) has two 150's as children and 580 (bottom) has three 150's as children. In this case the shorter of the two (top) goes one step right and then to the left child, and the longer of the two goes the child one beyond what they had in common. Observe that both structures to the left of the search path is identical. }
    \label{fig2}
\end{figure}
\end{landscape}

\begin{figure*}
    \centering
    \includegraphics[width=6in]{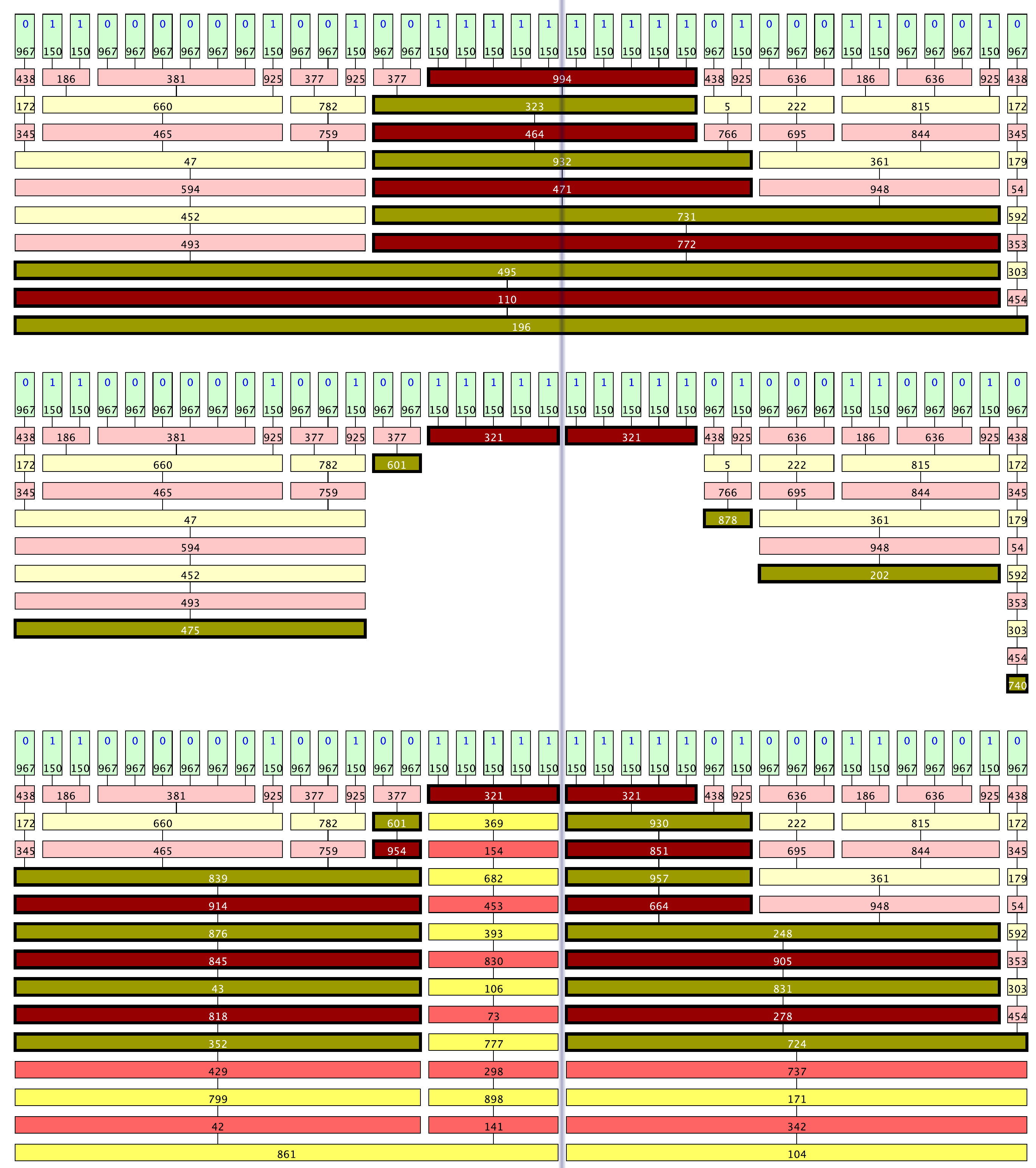}
    \caption{Split. In the first phase (top), those nodes that straddle the split point are identified. Then, they are split into zero, one or, two nodes, one which has children entirely to the left of the split point, and one which nodes children entirely to the right. This is illustrated in the middle figure, for example, 994 is split into two 321's, 323 is split into 601 on the left, and 464 is removed completely as its only child straddles the split line. Then the nodes without parents are incorporated into the structure in a bottom-up reconstruction. The dark nodes represent the fringe nodes, two per level, that delineate the area that was rebuilt, and only a logarithmic number of nodes are expected to be added in this process.}
    \label{fig3}
\end{figure*}

\clearpage

%\end{proof}

\bibliographystyle{alpha}
\bibliography{bib}

\end{document}